\documentclass[journal]{IEEEtran}

\usepackage[english]{babel}															% English language/hyphenation
\usepackage{graphics} % for pdf, bitmapped graphics files
\usepackage{mathtools} % assumes new font selection scheme installed
\usepackage{times} % assumes new font selection scheme installed
\usepackage{amsmath} % assumes amsmath package installed
\usepackage{amssymb}  % assumes amsmath package installed
\usepackage{amsthm}
\usepackage{amsfonts}%
\usepackage{euscript}
\usepackage{mathrsfs}
\usepackage[inline]{enumitem}
\usepackage{cite}
\usepackage[normalem]{ulem}
\newcommand{\NX}{{n_\mathrm{x}}}

\newcommand{\NY}{{n_\mathrm{y}}}

\newcommand{\NPSI}{{n_{\psi}}}			% Additional {} such that it can be used in super and subscript

\newcommand{\sU}{\mathbb{U}}
\newcommand{\sP}{\mathbb{P}}
\newcommand{\sX}{\mathbb{X}}
\newcommand{\sY}{\mathbb{Y}}
\newcommand{\sZ}{\mathbb{Z}}

\newcommand{\Ru}{\mathbb{R}^{n_\mathrm{u}}}
\newcommand{\Rp}{\mathbb{R}^{n_\mathrm{p}}}
\newcommand{\Rx}{\mathbb{R}^{n_\mathrm{x}}}
\newcommand{\Ry}{\mathbb{R}^{n_\mathrm{y}}}

\newcommand{\ind}[1]{{[#1]}}

\newcommand{\Qv}{\EuScript{Q}}
\newcommand{\Sv}{\EuScript{S}}
\newcommand{\Rv}{\EuScript{R}}

\newcommand{\expct}{\mathbb{E}}

\newcommand{\var}{\mathrm{var}}
\newcommand{\cov}{\mathrm{cov}}

\newcommand{\Tr}{{\mathrm{Tr}}}
\newcommand{\eye}{{\mathrm{I}}}

\newcommand{\Afnc}{\mathcal{A}}
\newcommand{\Bfnc}{\mathcal{B}}
\newcommand{\Cfnc}{\mathcal{C}}
\newcommand{\Dfnc}{\mathcal{D}}
\newcommand{\Gfnc}{\mathcal{G}}
\newcommand{\Hfnc}{\mathcal{H}}
\newcommand{\Kfnc}{\mathcal{K}}
\newcommand{\Pfnc}{\mathcal{P}}
\newcommand{\Rfnc}{\mathcal{R}}

\newtheorem{lem}{Lemma}
\newtheorem{thm}{Theorem}
\newtheorem{exmp}{Example}
\newtheorem{conj}{Conjecture}

\ifCLASSINFOpdf
\else
\fi

\begin{document}

\title{On the Connection Between Different Noise Structures for LPV-SS Models}

\author{Pepijn~B.~Cox and % ~\IEEEmembership{Member,~IEEE,}
        Roland~T\'{o}th%~\IEEEmembership{Member,~IEEE,}
%        and~Mih\'{a}ly~Petreczky,~\IEEEmembership{Member,~IEEE}% <-this % stops a space
\thanks{P.B. Cox and R. T\'{o}th are with the Control Systems Group, Department of Electrical Engineering, Eindhoven University of Technology, P.O. Box 513, 5600 MB Eindhoven, The Netherlands. E-mail: \{p.cox,r.toth\}@tue.nl.}% <-this % stops a space
%\thanks{M. Petreczky is with the Centre de Recherche en Informatique, Signal et Automatique de Lille (CRIStAL). Lille 1 University, M3, Avenue Carl Gauss, 59650 Villeneuve-d'Ascq, France. E-mail:mihaly.petreczky@ec-lille.fr.}% <-this % stops a space
%\thanks{Manuscript received September 5, 2016; revised August 26, 2015.}
}

%\markboth{IEEE TRANSACTIONS ON AUTOMATIC CONTROL,~VOL.~14, NO.~8, AUGUST~2015}%
%{Cox \MakeLowercase{\textit{et al.}}: On the Connection Between Different Noise Structures for LPV-SS Systems}

\maketitle

\begin{abstract}
Different representations to describe noise processes and finding connections or equivalence between them have been part of active research for decades, in particular for linear time-invariant case. In this paper the linear parameter-varying (LPV) setting is addressed; starting with the connection between an LPV state-space (SS) representation with a general noise structure and the LPV-SS model in an innovation structure, i.e., the Kalman filter. More specifically, the considered LPV-SS representation with general noise structure has static, affine dependence on the scheduling signal; however, we show that its companion innovation structure has a dynamic, rational dependency structure. Following, we would like to highlight the consequences of approximating this Kalman gain by a static, affine dependency structure. To this end, firstly, we use the ``fading memory'' effect of the Kalman filter to reason how the Kalman gain can be approximated to depend only on a partial trajectory of the scheduling signal. This effect is shown by proving an asymptotically decreasing error upper bound on the covariance matrix associated to the innovation structure in case the covariance matrix is subjected to an incorrect initialization or disturbance. Secondly, we show by an example that an LPV-SS representation that has dynamical, rational dependency on the scheduling signal can be transformed into static, affinely dependent representation by introducing additional states. Therefore, an approximated Kalman gain can, in some cases, be represented by a static, affine Kalman gain at the cost of additional states.
\end{abstract}


% \fi
\IEEEpeerreviewmaketitle

\section{Introduction}

%\pepijn{Position the paper from the identification point of view. In the LTI case there the connection between innovation an static gain well know. Many schemes use this ``simplified'' gain (projection SIDs). Based upon the LTI case, many extensions are known to linear systems with varying parameters. However, the question arises if these properties also hold for linear time-varying or linear parameter-varying systems. Answer comes in this paper}

Including general representations of noise processes in system identification is essential for capturing a wide variety of possible noise sources experienced in practice, e.g., unmodelled dynamics, sensor noise, parameter inaccuracies, etc. Hence, active research on different representations, their generality, and connections between them has been going on for decades. Especially, the \textit{linear time invariant} (LTI) case has a well established connection between an LTI \textit{state-space} (SS) representation with state and output additive noise and the innovation form, i.e., the Kalman filter (e.g., see~\cite{Anderson1979} and the references therein). In this case, the Kalman filter is asymptotically time invariant, therefore, a suboptimal filter can be found with a constant Kalman matrix.

To the authors knowledge, similar time invariant Kalman filters for \textit{linear parameter-varying} (LPV), time-varying, or nonlinear systems does not exists. Except~\cite{Petreczky2016}, for a stochastic jump-Markov linear system a Kalman gain that has affine dependency on the switching signal can be found, under the assumption that the system is quadratic stabilty. Contrary to the LTI case, having quadratic stability for LPV or jump-Markov linear systems is only a sufficient condition, i.e., ristrictive condition for stability (e.g., see~\cite{Boyd1989}), and, therefore, the result of~\cite{Petreczky2016} does not apply to every stable LPV or jump-Markov linear system. However, in this paper such strong assumption on stability of the system is not required. Hence, we will treat a more general case.

In this paper, we start by providing the connection between the LPV-SS representation with general noise and its LPV-SS innovation structure companion (Sec.~\ref{sec:preliminaries}). When moving from the general noise structure with affine, static dependency on the scheduling signal to the innovation form, the resulting error state covariance matrix function, Kalman gain, and the covariance of the innovation noise will have dynamic, rational dependency on the scheduling signal. Then, based upon the stability result of~\cite{Deyst1968}, we show that the innovation recursion can recover from an incorrect initialization or disturbance with a guaranteed asymptotic convergence (Sec.~\ref{sec:asymptoticConvergence}). This guaranteed convergence implies a ``fading memory'' effect within the innovation recursion and it is used to argue that the Kalman gain can be approximated by only using a partial trajectory of the scheduling signal, in stead of the complete trajectory (Sec.~\ref{sec:apprxKalmanGain}).

% ############################################### %
% ----------------------------------------------- %
% 		The LPV identification problem
% ----------------------------------------------- %
% ############################################### %
\section{Preliminaries} \label{sec:preliminaries}

% ####################################################
% Subsec: Notation
% ####################################################

\subsection*{Notation}
We denote a probability space as $(\Xi,\mathcal{F}_\Xi, \mathbf{P})$ where $\mathcal{F}_\Xi$  is the $\sigma$-algebra, defined over the sample space $\Xi$; and $\mathbf{P}:\mathcal{F}_\Xi\rightarrow[0,1]$ is the probability measure defined over the measurable space $(\Xi,\mathcal{F}_\Xi)$. Within this work, we consider random variables that take values on the Euclidean space. More precisely, for the given probability space  $(\Xi,\mathcal{F}_\Xi, \mathbf{P})$ we define a \textit{random variable} $\mathbf f$ as a measurable function $\mathbf f:\Xi\rightarrow\mathbb{R}^n$, which induces a probability measure on $(\mathbb{R}^n,\mathscr{B}(\mathbb{R}^n))$. As such, a realization $\nu\in\Xi$ of $\mathbf P$, denoted $\nu\sim\mathbf P$, defines a realization $f$ of $\mathbf f$, i.e., $f:=\mathbf f(\nu) $. Furthermore, a \textit{stochastic process} $\mathbf x$ is a collection of random variables $\mathbf x_t:\Xi\rightarrow\mathbb{R}^n$ indexed by the set $t\in\mathbb{Z}$ (discrete time), given as $\mathbf x=\{ \mathbf x_t : t\in\mathbb{Z} \}$. A realization $\nu\in\Xi$ of the stochastic process defines a signal trajectory $x:=\{ \mathbf x_t(\nu) : t\in\mathbb{Z} \}$. We call a stochastic process $\mathbf x$ \textit{stationary} if $\mathbf x_t$ has the same probability distribution on each time index as $\mathbf x_{t+\tau}$ for all $\tau\in\mathbb{N}$. %In addition, a stationary process consisting of uncorrelated random variables with zero mean and finite variance is called a \textit{white noise process}. 

In addition, the inequalities $A\succeq B$ and $A\succ B$, for two symmetric matrices $A$ and $B$ of equal dimension, imply that $A-B$ is semi-positive and positive definite, respectively.

% ####################################################
% Subsec: The data-generating system with general noise model
% ####################################################
\subsection{The LPV-SS representation with general noise} \label{subsec:Model}

Consider a \textit{multiple-input multiple-output} (MIMO), discrete-time linear parameter-varying data-generating system, defined by the following first-order difference equation, i.e., the LPV-SS representation with general noise model: %\vspace{-2mm}
\begin{subequations} \label{eq:SSrep}
\begin{alignat}{3}
  x_{t+1} &= \Afnc(p_t)&x_t&+\Bfnc(p_t)&&u_t+\Gfnc(p_t)w_t, \label{eq:SSrepState} \\
  y_t &= \Cfnc(p_t)&x_t&+\Dfnc(p_t)&&u_t+\Hfnc(p_t)v_t, \label{eq:SSrepOut}
\end{alignat}
\end{subequations} %\vskip -6mm \noindent 
where $x:\sZ\rightarrow\sX=\Rx$ is the state variable, $y:\sZ\rightarrow\sY=\Ry$ is the measured output signal, $u:\sZ\rightarrow\sU=\Ru$ denotes the input signal, $p:\sZ\rightarrow\sP \subseteq\Rp$ is the scheduling variable, subscript $t\in\mathbb{Z}$ is the discrete time, $w:\sZ\rightarrow\Rx$, $v:\sZ\rightarrow\Ry$ are the sample path realizations of zero-mean stationary noise processes: %\vspace{-2mm}
\begin{equation} \label{eq:noisy}
\left[ \begin{array}{c}
\mathbf w_t \\ \mathbf v_t
\end{array} \right] \sim \mathcal{N}(0,\Sigma),\hspace{1cm} \Sigma = \left[ \begin{array}{cc} \Qv & \Sv \\ \Sv^\top & \Rv \end{array} \right],
\end{equation} %\vskip -7mm \noindent
where $\mathbf w_t:\Xi\rightarrow\sX$, $\mathbf v_t:\Xi\rightarrow\sY$ are random variables of the stochastic process $\mathbf w$ ,$\mathbf v$, respectively, $\Qv\in\mathbb{R}^{\NX\times\NX}$, $\Sv\in\mathbb{R}^{\NX\times\NY}$, and $\Rv\in\mathbb{R}^{\NY\times\NY}$ are covariance matrices, such that $\Sigma$ is positive definite. Furthermore, we will assume $u,p,w,v,y$ to have left compact support to avoid technicalities with initial conditions. The matrix functions $\Afnc(\cdot),...,\Hfnc(\cdot)$, defining the SS representation \eqref{eq:SSrep} are defined as affine combinations: %\vspace{-5mm}
\begin{equation}\label{eq:sysMatrices}
\begin{aligned} 
 \Afnc(p_t)&\!=\!A_0+\hspace{-1mm}\sum_{i=1}^{\NPSI} A_i\psi^\ind{i}(p_t),\hspace{-1mm}&\Bfnc(p_t)&\!=\!B_0+\hspace{-1mm}\sum_{i=1}^{\NPSI} B_i\psi^\ind{i}(p_t) , \\
 \Cfnc(p_t)&\!=\!C_0+\hspace{-1mm}\sum_{i=1}^{\NPSI} C_i\psi^\ind{i}(p_t),\hspace{-1mm}&\Dfnc(p_t)&\!=\!D_0+\hspace{-1mm}\sum_{i=1}^{\NPSI} D_i\psi^\ind{i}(p_t), \\
 \Gfnc(p_t)&\!=\!G_0+\hspace{-1mm}\sum_{i=1}^{\NPSI} G_i\psi^\ind{i}(p_t),\hspace{-1mm}&\Hfnc(p_t)&\!=\!H_0+\hspace{-1mm}\sum_{i=1}^{\NPSI} H_i\psi^\ind{i}(p_t), \\[-2mm]
\end{aligned}
\end{equation}
where $\psi^\ind{i}(\cdot):\sP\rightarrow\mathbb{R}$ are bounded scalar functions on $\sP$ and $\{A_i,B_i,C_i,D_i,G_i,H_i\}_{i=0}^{\NPSI}$  are constant matrices with appropriate dimensions. Additionally, for well-posedness, it is assumed that $\{\psi^\ind{i}\}_{i=1}^{\NPSI}$ are  linearly independent over an appropriate function space and are normalized w.r.t. an appropriate norm or inner product~\cite{Toth2012}.

% ####################################################
% Subsec: LPV-SS innovation from
% ####################################################

\subsection{The innovation form}

 To start, under some mild conditions, the LPV-SS representation~\eqref{eq:SSrep} has the following equivalent innovation form:

%\vspace{-2mm}
\begin{lem} \label{lem:inn} For each given trajectory of the input $u$ and scheduling $p$, the LPV data-generating system~\eqref{eq:SSrep} can be equivalently represented by a $p$-dependent innovation form %\vspace{-2.5mm}
\begin{subequations} \label{eq:SSrepInn}
\begin{alignat}{3}
  \check x_{t+1} &= \Afnc(p_t)&\check x_t&+\Bfnc(p_t)&&u_t+\Kfnc_t\xi_t, \label{eq:SSrepInnState} \\
  y_t &= \Cfnc(p_t)&\check x_t&+\Dfnc(p_t)&&u_t+\xi_t, \label{eq:SSrepInnOut}
\end{alignat} % \vskip -6mm \noindent
where $\boldsymbol\xi_t\sim\mathcal{N}(0,\Omega_t)$ and $\Kfnc_t$ can be uniquely determined by %\vspace{-3mm}
\begin{align}
\Kfnc_t\! &=\! \left[ \Afnc(p_t) \Pfnc_{t\vert t-1} \Cfnc^\top\!\!(p_t) + \Gfnc(p_t)\Sv\Hfnc^\top\!\!(p_t)\right] \Omega_t^{-1}\!\!, \label{eq-lem:kalmanFilter} \\
\Pfnc_{t+1\vert t}\! &=\! \Afnc(p_t) \Pfnc_{t\vert t-1} \Afnc^\top\!\!(p_t)  - \Kfnc_t \Omega_t \Kfnc^\top_t +\nonumber \\ & \quad\quad \Gfnc(p_t) \Qv \Gfnc^\top\!\!(p_t), \label{eq-lem:innErrorCov} \\
\Omega_t\! &=\! \Cfnc(p_t)\Pfnc_{t\vert t-1}\Cfnc^\top\!\!(p_t) + \Hfnc(p_t)\Rv\Hfnc^\top\!\!(p_t), \label{eq-lem:noiseCov}
\end{align} %\vskip -6.5mm \noindent
\end{subequations}
under the assumption that $\exists t_0\in\mathbb{Z}$ such that $x_{t_0}=0$ and $\Omega_t$ is non-singular for all $t\in[t_0,\infty)$. In~\eqref{eq-lem:kalmanFilter}-\eqref{eq-lem:noiseCov}, the notation of $\Kfnc_t$, $P_{t+1\vert t}$, and  $\Omega_t$ is a shorthand for $\Kfnc_t\coloneqq(\Kfnc\diamond p_t)\in\mathscr{R}^{\NX\times\NY}$, $\Pfnc_{t+1\vert t}\coloneqq(P_{t+1\vert t}\diamond p_t)\in\mathscr{R}^{\NX\times\NX}$, and  $\Omega_t\coloneqq(\Omega\diamond p_t)\in\mathscr{R}^{\NY\times\NY}$. The operator $\diamond:(\mathscr{R},\sP^\sZ)\rightarrow\mathbb{R}^\sZ$ denotes $(\Kfnc_t\diamond p_t)=\Kfnc_t(p_{t+\tau_1},\ldots,p_t,\ldots,p_{t-\tau_2})$ with $\tau_1,\tau_2\in\mathbb{Z}$. The subscript notation $_{t+1\vert t}$ denotes that the matrix function at time ${t+1}$ depends only on $p_i$ for $i=t_0,\ldots,t$. \hfill $\blacksquare$
\end{lem} %\vskip -6mm \noindent
\begin{proof}
%The scheduling signal is independent of the noise in the innovation representation~\eqref{eq:SSrepInn}, hence, the innovation system is directly related to a time-varying Kalman filter, e.g., see~\cite{Anderson1979}.
See Appendix.
\end{proof} %\vskip -4mm \noindent

From Lem.~\ref{lem:inn} it becomes clear that moving from the LPV-SS system~\eqref{eq:SSrep} with static, affine dependency to the innovation from comes at the cost of dynamic, rational dependency on the scheduling signal. In Sec.~\ref{sec:asymptoticConvergence}, guaranteed asymptotic convergence of the covariance matrix $\Pfnc_{t+1\vert t}$~\eqref{eq-lem:innErrorCov} is proven if the covariance matrix is perturb by an error in the past. That result is used to argue how the Kalman gain $\Kfnc_t$~\eqref{eq-lem:kalmanFilter} can be approximated by a partial trajectory of the scheduling signal, in Sec.~\ref{sec:apprxKalmanGain}.

% ############################################### %
% ----------------------------------------------- %
% 		Properties of the Kalman gain
% ----------------------------------------------- %
% ############################################### %

\section{Garuenteed Asymptotic Convergence of the Covariance Matrix $\Pfnc_{t+1\vert t}$} \label{sec:asymptoticConvergence}

In this section, we will show that an error created on the priori error covariance matrix $\Pfnc_{t+1\vert t}$~\eqref{eq-lem:innErrorCov} at a certain time will asymptotically decrease to zero when time progresses. %This upper-bound on the error is used in Sec.~\ref{sec:apprxKalmanGain} to argue how the Kalman gain $\Kfnc_t$~\eqref{eq-lem:kalmanFilter} can be approximated by only using a part of the scheduling signals' trajectory.

To this end, let us introduce some technicalities. Firstly, the stochastic processes $\mathbf w$ and $\mathbf v$~\eqref{eq:noisy} need to be uncorrelated, hence using the minimum variance estimate of $\mathbf w_t$ given by $\mathbf{\bar w}_t = \mathbf w_t-\Sv\Rv^{-1}\mathbf v_t$, the state equation~\eqref{eq:SSrepState} is rewritten as (e.g., see~\cite[Section 5.5]{Anderson1979}) %\vspace{-2mm}
\begin{subequations} \label{eq:SSrepTrans}
\begin{align}
x_{t+1} = &\left(\Afnc(p_t) - \Gfnc(p_t)\Sv\Rv^{-1}\Hfnc^{-1}(p_t)\Cfnc(p_t)  \right)x_t \nonumber \\
		  & + \left( \Bfnc(p_t)- \Gfnc(p_t)\Sv\Rv^{-1}\Hfnc^{-1}(p_t)\Dfnc(p_t) \right) u_t  \nonumber \\ 
		  & + \Gfnc(p_t)\Sv\Rv^{-1}\Hfnc^{-1}(p_t) y_t +\Gfnc(p_t) \bar w_t, \label{eq:TVrepresentationStateInd}  
\end{align} %\vskip -8mm \noindent
where %\vspace{-3mm}
\begin{equation} \label{eq:LPVrepresentationNoiseInd}
\left[\begin{array}{c} \mathbf{\bar w}_t\\ \mathbf v_t \end{array} \right] \sim \mathcal{N}\left( \left[\begin{array}{c} 0 \\ 0 \end{array}\right],\left[ \begin{array}{cc} \Qv-\Sv\Rv^{-1}\Sv & 0  \\ 0 & \Rv \end{array}  \right] \right).
\end{equation} %\vskip -7mm \noindent
Define $\bar u^\top_t \coloneqq [\begin{array}{cc} u^\top_t & y^\top_t \end{array}]^\top$, which gives the following scheduling dependent matrices %\vspace{-2.5mm}
\begin{align}
\bar \Afnc_t  &\!=\! \Afnc(p_t) - \Gfnc(p_t)\Sv\Rv^{-1}\Hfnc^{-1}(p_t)\Cfnc(p_t), \\ 
\bar \Bfnc_t  &\!=\! \Big[  \Bfnc(p_t)- \Gfnc(p_t)\Sv\Rv^{-1}\Hfnc^{-1}(p_t)\Dfnc(p_t), \nonumber \\
			  & \qquad\qquad   \Gfnc(p_t)\Sv\Rv^{-1}\Hfnc^{-1}(p_t) \Big]\!, \\
\bar \Qv  & \!=\! \Qv-\Sv\Rv^{-1}\Sv.
\end{align} %\vskip -6.5mm \noindent
Secondly, let $\bar\Bfnc_t$ and $\Dfnc(p_t)$ be bounded and assume that %\vspace{-2.5mm}%, in the sequel, assume that there exists real numbers $\alpha_1,\alpha_2,\beta_1,\beta_2,\gamma_1,\gamma_2$ such that
\begin{align}
\alpha_1 I &\succeq \Gfnc(p_i) \bar\Qv \Gfnc(p_i) \succeq \alpha_2 I, \label{eq:stochasticCtrb} \\
\beta_1 I &\preceq \Cfnc(p_i)^\top(\Hfnc(p_i)\Rv\Hfnc(p_i)^\top)^{-1}\Cfnc(p_i) \preceq \beta_2I, \label{eq:stochasticObsv} \\
\delta_1 I &\succeq \bar\Afnc^\top_i \bar\Afnc_i \succeq \delta_2 I \label{eq:boundedInvertA},
\end{align} %\vskip -7mm \noindent
holds for $i\in\mathbb{T}$, with left compact support $\mathbb{T}$ of the scheduling signal, $\alpha_1,\beta_2,\delta_1>0$, and $\alpha_2,\beta_1,\delta_2<\infty$. Conditions~\eqref{eq:stochasticCtrb}-\eqref{eq:boundedInvertA} imply that the system~\eqref{eq:SSrep} is stochastically controllable and observable for all possible variations of $p\in\sP$, i.e., the state can uniquely be reconstructed, which are not over restrictive assumptions. Define the \textit{posterior} filter error $e_t=x_t-\check x_{t\vert t}$ dynamics by \vspace{-2mm}
\begin{equation} \label{eq:error}
e_{t+1} = (I-\Kfnc_{t+1}\Cfnc(p_{t+1}))\bar\Afnc_te_t = W_t e_t.
\end{equation}%\vskip -4mm \noindent
If~\eqref{eq:stochasticCtrb}-\eqref{eq:boundedInvertA} hold then there exitst a $\beta_3,\beta_5,\beta_6$~\cite{Deyst1968} as %\vspace{-3mm}
\begin{equation}
\!\left\Vert \! \left[ \! \begin{array}{cc} \Cfnc(p_i) & 0 \\ 0 & \Cfnc(p_{i-1}) \end{array}\! \right]\!\!\left[ \begin{array}{c}  \bar\Afnc_{i-1} \bar\Afnc_{i-2} \\ \Afnc_{i-1} \end{array} \right] e_{i-2}  \right\Vert_2 \!\!\!\geq\! \beta_3\! \left\Vert e_{i-2} \right\Vert_2\!,
\end{equation} %\vskip -6mm \noindent
%and
\begin{equation}
\left[ \begin{array}{cc} \Rv & 0 \\ 0 & \Rv \end{array} \right] + L_i \left[ \begin{array}{cc} \Pfnc_{i\vert i-1} & 0 \\ 0 & \Pfnc_{i-1\vert i-2} \end{array} \right] L_i^\top \preceq \beta_5 I,
\end{equation} %\vskip -2mm \noindent
with %\vspace{-2mm}
\begin{equation*}
L_i = \left[ \begin{array}{cc} \Cfnc(p_i) & \Cfnc(p_i) \bar\Afnc_{i-1} \\ 0 & \Cfnc_{i-1} \end{array} \right],
\end{equation*} %\vskip -2mm \noindent
and %\vspace{-1mm}
\begin{equation}
\left\Vert W^{-1}_{i-2} W^{-1}_{i-1} e_i \right\Vert_2 \geq \beta_6 \left\Vert e_i \right\Vert_2, \label{eq:cntBeta6}
\end{equation} %\vskip -3mm \noindent
for $i=\mathbb{T}$ where $0<\beta_3,\beta_5,\beta_6<\infty$. Using~\eqref{eq:stochasticCtrb}-\eqref{eq:cntBeta6}, the error created on the priori error covariance matrix $\Pfnc_{t+1\vert t}$~\eqref{eq-lem:innErrorCov} at a certain time will asymptotically decrease to zero as:
\end{subequations}

\begin{thm} \label{thm:boundErrorCov}
\begin{subequations}
Let $\bar\Bfnc_t$ and $\Dfnc(p_t)$ be bounded and~\eqref{eq:stochasticCtrb}-\eqref{eq:boundedInvertA} hold. Assume that $\Pfnc_{t\vert t-1}$ is a positive semi definite matrix function of $t$, which has a left compact support $\mathbb{T}$ and satisfies~\eqref{eq-lem:kalmanFilter}-\eqref{eq-lem:noiseCov} for a given trajectory of $p$ (with left compact support). Let us consider that there exists a $\tau>0$ and an associated $\hat\Pfnc^{(\tau)}_{t\vert t-1}$ constructed with the recursions of~\eqref{eq-lem:innErrorCov} initialized at time $t-\tau$ with %\vspace{-1mm}
\begin{equation} \label{eq-thm:initErrCov}
\hat\Pfnc^{(\tau)}_{t-\tau\vert t-\tau-1} = \bar\Afnc_t\hat P_{0,\tau}\bar\Afnc_t^\top + \bar \Qv
\end{equation}%\vskip -7mm \noindent
where $\hat P_{0,\tau}$ is a static matrix defined as %\vspace{-2mm}
\begin{equation} \label{eq-thm:condErrorCov}
0\prec\hat P_{0,\tau} \prec \left( \frac{\alpha_2}{\alpha_2\beta_2+1} \right) I.
\end{equation} %\vskip -3mm \noindent
Then, the difference between $\Pfnc_{t\vert t-1}$ and $\hat\Pfnc^{(\tau)}_{t\vert t-1}$ has the following bound %\vspace{-2mm}
\begin{equation} \label{eq-thm:maxCovarianceBound}
\max_{t\in\mathbb{T}}\!\left\Vert\Pfnc_{t\vert t-1}\!-\!\hat \Pfnc^{(\tau)}_{t\vert t-1}\right\Vert_2  \!\!\leq \xi^\tau \frac{\delta_1\NX(\alpha_1\beta_1\!+\! 1)^2 (\alpha_2\beta_2\! +\! 1)}{\alpha_2\beta_1^2}\!,
\end{equation} %\vskip -3mm \noindent
where $\xi\in(0,1)$ and given as %\vspace{-3mm}
\begin{equation*}
\xi  = \frac{\beta_5(\alpha_2\beta_2+1)}{\beta_5(\alpha_2\beta_2+1) + \alpha_2\beta_3\beta_6}
\end{equation*} \vskip -4mm \noindent \hfill $\blacksquare$
\end{subequations}
\end{thm} 
%\begin{proof}
%See Appendix~\ref{app-sec:ErrorCovar}. \hfill $\blacksquare$
%\end{proof} 
%%

The remainder of the section is used to proof Theorem~\ref{thm:boundErrorCov}. The following proof uses extensivly the result of~\cite{Deyst1968}, however, Deyst and Price make use of the \textit{posterior} covariance matrix $\Pfnc_{t\vert t}$ in stead of the \textit{prior} covariance matrix $\Pfnc_{t\vert t-1}$. Hence, we will first construct the proof w.r.t. posterior covariance matrix. The covariance matrix is given as
\begin{equation} \label{app-eq:postErrCov}
\Pfnc_{t\vert t}= \left(\eye-\Kfnc_t\Cfnc(p_t)\right) \Pfnc_{t\vert t-1},
\end{equation}
where it is proven in~\cite{Deyst1968} to be bounded as
\begin{equation} \label{eq-app:boundPostErrCov}
\left( \frac{\alpha_2}{\alpha_2\beta_2+1} \right) I\preceq\Pfnc_{t\vert t}\preceq \left(\frac{1}{\beta_1}+\alpha_1\right)I,
\end{equation}
and the prior covariance can be found from the posterior as
\begin{equation} \label{eq-app:postToPriorErrCov}
\Pfnc_{t+1\vert t} = \bar\Afnc_t\Pfnc_{t\vert t}\bar\Afnc_t^\top + \bar \Qv.
\end{equation}
Remark that $\hat P_{0,\tau}$ of~\eqref{eq-thm:initErrCov} substitutes the posterior covariance matrix on $\Pfnc_{t-\tau\vert t-\tau}$ to construct $\Pfnc^{\scriptstyle(\tau)}_{t\vert t}$.
\begin{lem} \label{lem-app:inn}
Given~\eqref{eq-thm:condErrorCov} and the error dynamics~\eqref{eq:error}, it holds that
\begin{multline} \label{eq-app:upperPostError}
0 \preceq \left(\Pfnc_{t\vert t}-\hat \Pfnc^{(\tau)}_{t\vert t}\right) \\ \! \preceq\!\! \prod_{i=t-\tau}^{t-1}\hat W_i\left(\Pfnc_{t-\tau\vert t-\tau}-\hat P_{0,\tau}\right)\prod_{i=t-\tau}^{t-1}\hat W_i^\top,
\end{multline}
where $\hat W_i$ is the filter error dynamics w.r.t. $\hat \Kfnc_{i+1}$ of $\hat \Pfnc^{\scriptstyle(\tau)}_{i\vert i}$. \hfill $\blacksquare$
\end{lem}
\begin{proof}
\begin{subequations}
Lets first proof the lemma for $\tau=1$. As $\Pfnc_{t\vert t}$ is the optimal solution for $\Kfnc_t$, then $\Pfnc_{t\vert t}'$ is constructed from $\Kfnc_i$ for $i=t_0,\ldots,t-\tau-1$ and $\Kfnc^{\scriptstyle(\tau)}_{j}$ for $j=t-\tau,\ldots,t$ with $\Kfnc_j\neq \Kfnc^{\scriptstyle(\tau)}_j$. Remark that $\Pfnc_{t\vert t}'$ is suboptimal; hence, $\Pfnc_{t\vert t}'\succeq \Pfnc_{t\vert t}$~\cite[Theorem 2.1]{Anderson1979}. Therefore,
\begin{equation} \label{eq-app:inequalityCov}
 \Pfnc_{t\vert t}'-\hat \Pfnc^{(1)}_{t\vert t}  \succeq \Pfnc_{t\vert t}-\hat \Pfnc^{(1)}_{t\vert t}.
\end{equation}
Also see that~\eqref{app-eq:postErrCov} can be written, by using~\eqref{eq-lem:innErrorCov}, as
\begin{align}
\Pfnc_{t\vert t} \! =& \! \left(I\!-\!\Kfnc_t\Cfnc(p_t)\!\right)\!\!\left( \bar\Afnc_{t-1} \Pfnc_{t-1\vert t-1}  \bar\Afnc^\top_{t-1}\!+\!\Gfnc(p_{t-1}) \bar\Qv\Gfnc^\top\!\!(p_{t-1})\! \right)\nonumber\\ 
				 & \quad\left(I-\Kfnc_t\Cfnc(p_t)\right)^\top  + \Kfnc_{t-1}\Rv\Kfnc_{t-1}^\top . \label{eq-app:prioriCovariance}
\end{align}
Combining~\eqref{eq-app:inequalityCov} and~\eqref{eq-app:prioriCovariance} results in
\begin{multline} \label{eq-app:upper1stepBound}
\hat W_{t-1}\! \left(\Pfnc_{t-1\vert t-1}-\hat P_{0,\tau} \right)\! \hat W^\top_{t-1} \succeq \Pfnc_{t\vert t}-\hat \Pfnc^{(1)}_{t\vert t}\!.
\end{multline}
Then repeating the upper bound~\eqref{eq-app:upper1stepBound} for $\tau$ time steps proofs the upper bound, i.e., right-hand side of~\eqref{eq-app:upperPostError}. \\
Similar argument can be made for the lower bound. Now, initialize with $\hat P_{0,\tau}$ and use $\Kfnc_j$ for $j=t-\tau,\ldots,t$ to construct $\hat \Pfnc_{t\vert t}'$, i.e., find an suboptimal solution with $\hat \Pfnc_{t\vert t}'\succeq\hat \Pfnc^{\scriptstyle(\tau)}_{t\vert t}$. Hence $\Pfnc_{t\vert t}-\hat \Pfnc^{\scriptstyle(\tau)}_{t\vert t}  \succeq \Pfnc_{t\vert t}-\hat \Pfnc_{t\vert t}'$, which results, for $\tau=1$, in the following lower bound
\begin{multline} \label{eq-app:lower1stepBound}
W_{t-1}\!\left(\Pfnc_{t-1\vert t-1}-\hat P_{0,1} \right)\! W^\top_{t-1} \preceq \Pfnc_{t\vert t}-\hat \Pfnc^{(1)}_{t\vert t}.
\end{multline}
\end{subequations}
As $\Pfnc_{t-1\vert t-1}-\hat P_{0,\tau}$ is semi-positive definite (by construction of $\hat P_{0,\tau}$), the left-hand side of~\eqref{eq-app:lower1stepBound} is bounded by zero. 
\end{proof}
Next, let us provide the sufficient conditions for quadratic Lyapunov stability of the filter dynamics~\eqref{eq:error} proven in~\cite{Deyst1968}:
\begin{lem} \label{lem-app:asymptoticallyStability}
\begin{subequations}
If the LPV-SS system~\eqref{eq:SSrep} satisfies conditions~\eqref{eq:stochasticCtrb}-\eqref{eq:boundedInvertA} then the system~\eqref{eq:error} is asymptotically stable. Additionally, there exists a real scalar function $V(e_t,t)$ such that
\begin{align}
0<\gamma_1\Vert e_t\Vert^2_2\leq V(e_t,t)&\leq\gamma_2\Vert e_t\Vert^2_2, e_t\neq0, \label{eq-app:boundedLya} \\
V(e_t,t) - V(e_{t-1},t-1) &\leq \gamma_3\Vert e_t\Vert^2_2 <0, e_t\neq0, \label{eq-app:boundedDif}
\end{align}
where
\begin{equation} \label{eq-app:constGam}
\gamma_1\!=\!\frac{\beta_1}{1\!+\!\alpha_1\beta_1},\hspace{3mm}\gamma_2 \!=\! \frac{1}{\alpha_2}\!+\beta_2,\hspace{3mm}\gamma_3 \!=\! -\beta_3^2\beta_5^{-1}\beta_6.
\end{equation}
\end{subequations}
\hfill $\blacksquare$
\end{lem}
%
%The bounds~\eqref{eq-thm:stochasticCtrb}-\eqref{eq-thm:boundedInvertA} are sufficient conditions for the quadratic Lyapunov stability of Lemma~\ref{app-prop:asymptoticallyStability}. As with these bounds one can proof boundedness on the estimation error covariance matrices $P_{t\vert t}$ and $P_{t\vert t-1}$ of the minimum variance estimator for the TV-SS system~\eqref{eq:TVrepresentation} and, therefore, quadratic Lyapunov stability on the error dynamics. 
Using Lem.~\ref{lem-app:asymptoticallyStability}, the bound on the error-dynamics is:
\begin{lem} \label{lem-app:lemBoundDyn}
If the LPV-SS system~\eqref{eq:SSrep} satisfies conditions~\eqref{eq:stochasticCtrb}-\eqref{eq:boundedInvertA} then the $\tau$-step homogeneous error dynamics~\eqref{eq:error} are bounded as
\begin{equation} \label{eq-app:lemBoundDyn}
	\Vert e_t \Vert^2_2 = \left\Vert \prod_{i=t-\tau}^{t-1}W_i e_{t-\tau} \right\Vert^2_2 \leq \xi^\tau\gamma_1^{-1}\gamma_2 \Vert e_{t-\tau}\Vert^2_2,
\end{equation}
where $\xi = \frac{\gamma_2}{\gamma_2-\gamma_3}$. \hfill $\blacksquare$
\end{lem}
\begin{proof}
\begin{subequations}
To simplify notation, define $V_t\coloneqq V(e_t,t)$. Substituting the upper bound of~\eqref{eq-app:boundedLya} into~\eqref{eq-app:boundedDif} gives
\begin{equation}
V_t - V_{t-1} \leq \gamma_3 \Vert e_t \Vert^2_2 \leq \gamma_3 \gamma_2^{-1} V_t < 0.
\end{equation}
Therefore the following holds
\begin{equation} \label{eq-app:boundDifLya}
V_t - \xi V_{t-1} \leq 0, \hspace{0.5cm} \xi = \frac{1}{1-\gamma_3\gamma_2^{-1}},
\end{equation}
where $\xi\in[0,1)$ as $\gamma_3\gamma_2^{-1}<0$. Hence, $V_t-\xi^\tau V_{t-\tau}\leq0$. Combining this $\tau$-step Lyapunov bound and~\eqref{eq-app:boundedLya} gives
\begin{equation}\label{eq-app:NstepBound}
\gamma_1 \Vert e_t\Vert^2_2 \leq V_t \leq \xi^\tau\gamma_2\Vert e_{t-\tau} \Vert^2_2
\end{equation}
\end{subequations}
The proof is completed by applying the $\ell_2$ norm on~\eqref{eq:error} and substituting~\eqref{eq-app:NstepBound}.
\end{proof}
To complete the proof of Thm.~\ref{thm:boundErrorCov}. First, take the eigenvalue decomposition $Z=UDU^\top$, where $U$ is a matrix containing the eigenvectors and $D$ the diagonal matrix containing the eigenvalues $\lambda_i\geq0$ of $Z$. For a positive definite matrix $Z$, it holds that  %As the eigenvalues are non-negative, it holds that $\Tr\left(D\right)I\succeq D$. Hence, it holds that 
\begin{equation} \label{eq-app:BoundByTrace}
Z = UDU^\top \preceq \Tr\left( D \right) UU^\top = \Tr\left( Z \right) I.
\end{equation}
%
%
%So, take the eigenvalue decomposition of $P_{t\vert t}-\hat P_{t\vert t} = \sum_{j=1}^\NX \lambda^2_j u_ju_j^\top$. Then the trace of the right-hand side of~\eqref{eq-app:upperPostError} is given by
%\begin{equation*}
%\Tr\left[ \prod_{i=t}^{t+N}\hat W_i\left(\sum_{j=1}^\NX \lambda^2_j u_ju_j^\top\right) \prod_{i=t}^{t+N}\hat W_i^\top \right].
%\end{equation*}
%
Second, take the eigenvalue decomposition of $\Pfnc_{t\vert t}-\hat \Pfnc_{t\vert t} = \sum_{j=1}^\NX \lambda^2_{j,t} u_{j,t}u_{j,t}^\top$ and define $y_{j,t}\coloneqq\lambda_{j,t} u_{j,t}$. Then, combining the trace of the left-hand side of~\eqref{eq-app:upperPostError} with~\eqref{eq-app:lemBoundDyn} gives
\begin{multline}
\!\!\sum_{j=1}^\NX \Tr\left[ y_{j,t}y^\top_{j,t} \right] = \sum_{j=1}^\NX \Vert y_{j,t}\Vert^2_2 \!\leq \!\sum_{j=1}^\NX \xi^\tau \gamma_1^{-1}\gamma_2 \Vert y_{j,t-\tau}\Vert^2_2 \\
\!=\xi^\tau \gamma_1^{-1}\gamma_2\Tr\!\left[ \Pfnc_{t-\tau\vert t-\tau}\!-P_{0,\tau}\right]\!\!. \label{eq-app:BoundtraceUpperBound}
\end{multline}
Joining~\eqref{eq-app:BoundByTrace} and~\eqref{eq-app:BoundtraceUpperBound} results in
\begin{equation*}
\Pfnc_{t\vert t}-\hat \Pfnc^{(\tau)}_{t\vert t} \preceq\xi^\tau \gamma_1^{-1}\gamma_2\Tr\left[ \Pfnc_{t-\tau\vert t-\tau}-P_{0,\tau} \right] I.
\end{equation*}
Left and right multiplying by $\bar\Afnc_t$ and $\bar\Afnc_t^\top$, respectively, and substituting~\eqref{eq-app:postToPriorErrCov} gives
\begin{multline}
\bar\Afnc_t\left(\Pfnc_{t\vert t}-\hat \Pfnc^{(\tau)}_{t\vert t}\right)\bar\Afnc_t^\top = \Pfnc_{t+1\vert t}-\hat \Pfnc^{(\tau)}_{t+1\vert t} \\  
\preceq\xi^\tau \gamma_1^{-1}\gamma_2\Tr\left[ \Pfnc_{t-\tau\vert t-\tau}-\hat P_{0,\tau} \right] \bar\Afnc_t\bar\Afnc_t^\top \\ 
\preceq\xi^\tau \gamma_1^{-1}\gamma_2\delta_1\Tr\left[ \Pfnc_{t-\tau\vert t-\tau}-\hat P_{0,\tau} \right] I. \label{eq-app:BoundtraceUpperBoundProiriToPostiroiri}
\end{multline}
%The upper bound of~\eqref{eq-app:BoundtraceUpperBoundProiriToPostiroiri} contains the a posteriori error covariance matrix. So, third, as $\bar\Afnc_i^\top \bar\Afnc_i$ is bounded~\eqref{eq-thm:boundedInvertA}, see that
%\begin{multline}
%\Tr\left[ \Pfnc_{i+1\vert i}-\hat \Pfnc_{i+1 \vert i} \right] = \Tr\left[ \bar\Afnc_i\left( \Pfnc_{i\vert i}-\hat \Pfnc_{i\vert i}\right) \bar\Afnc_i^\top \right] \\
%\succeq \sum_j^{\NX} \lambda_j(\bar\Afnc_i\bar\Afnc_i^\top)\lambda_{\NX-j+1}(\Pfnc_{i\vert i}-\hat \Pfnc_{i\vert i}) \\ 
%\succeq \delta_2  \Tr\left[ \Pfnc_{i\vert i}-\hat \Pfnc_{i\vert i} \right]. \label{eq-app:inequlityOnCovarianceProiriToPostiroiri}
%\end{multline}
%%
%The second line of~\eqref{eq-app:inequlityOnCovarianceProiriToPostiroiri} is given in~\cite[Theorem H.1.i]{Marshall2011}. Substituting~\eqref{eq-app:inequlityOnCovarianceProiriToPostiroiri} in~\eqref{eq-app:BoundtraceUpperBoundProiriToPostiroiri} will conclude the proof.
Taking into account~\eqref{eq-app:boundPostErrCov} and using~\eqref{eq-thm:condErrorCov}, the following holds
\begin{equation} \label{eq-app:boundTrace}
\Tr\left[ \Pfnc_{t-\tau\vert t-\tau}-\hat P_{0,\tau} \right] < \NX (\beta_1^{-1}+\alpha_1).
\end{equation}
To conclude the proof, the spectral norm of a matrix $A$ is $\Vert A \Vert_2=\sigma_{\mathrm{max}}(A)$. Hence, applying the spectral norm on~\eqref{eq-app:BoundtraceUpperBoundProiriToPostiroiri} and substitute~\eqref{eq-app:boundTrace} results in
\begin{equation}
\left\Vert\Pfnc_{t\vert t-1}-\hat \Pfnc^{(\tau)}_{t\vert t-1}\right\Vert_2  \!\leq\! \xi^{\tau} \gamma_1^{-1}\gamma_2\delta_1\NX (\beta_1^{-1}\!+\!\alpha_1),
\end{equation}
which is equivalent to~\eqref{eq-thm:maxCovarianceBound} when substituting~\eqref{eq-app:constGam} and taking into account that the bound is time independent, i.e., it should hold for every $t\in\mathbb{T}$, which concludes the proof of Thm.~\ref{thm:boundErrorCov}.
%
%where $\Vert A \Vert_2$ is the spectral norm of matrix $A$ or
%\begin{equation}
%\Vert\Pfnc_{t\vert t-1}-\hat \Pfnc_{t\vert t-1}\Vert_F  < \xi^N \gamma_1^{-1}\gamma_2\delta_1\delta_2^{-1}\NX^{3/2} (\beta_1^{-1}+\alpha_1)
%\end{equation}
%
%\hfill $\blacksquare$

% ############################################### %
% ----------------------------------------------- %
% 		Approximation of Kalman gain
% ----------------------------------------------- %
% ############################################### %

\section{Approximation of Kalman gain} \label{sec:apprxKalmanGain}

The innovation form is a different view on constructing a the Kalman filter for~\eqref{eq:SSrep}. Hence, $\mathcal{K}_t$ in~\eqref{eq-lem:kalmanFilter} can be viewed as the optimal LPV Kalman gain of~\eqref{eq:SSrep}. In the LTI case, the Kalman filter is asymptotically time invariant, therefore, a suboptimal filter can be found with a constant $P$ and $K$ matrix~\cite{Anderson1979}. Hence, in the LTI case, the innovation form with constant $K$ and $P$ matrix is viewed as a model description which allows a general noise model. However, for the LPV case, Lem.~\ref{lem:inn} indicates that even if $\Afnc(\cdot),\ldots,\Dfnc(\cdot)$ have, for example, affine dependence on $p_t$ (each $\psi^\ind{i}(p_t)=p^\ind{i}_t$) then $K_t,P_{t\vert t-1},\Omega_t$ are meromorphic functions, where the nominator and denominator are polynomial functions in the scheduling signal $p$ and its past time-shifts. Hence, the filter, generally speaking, it is not clear that $\Kfnc_t$ will converge to a steady state solution with some constant $K$ matrix, and, therefore, $\Kfnc_t$ is a function of scheduling signal and its past, i.e., $p_i$ with $i\in\mathbb{T}$.

However, a popular model for many subspace identification schemes is the innovation form, e.g., see~\cite{Verhaegen2007}. In the LTI case, the connection between the innovation form and the LTI counterpart of~\eqref{eq:SSrep}, e.g., $\Afnc(p)=A$, is well studied. However, it has not been thoroughly investigated in the LPV case. As Lem.~\ref{lem:inn} shows, the LPV-SS representation with general noise~\eqref{eq:SSrep} is not equivalent to the innovation form with only static, affine matrix functions, commonly used~\cite{Verdult02,Wingerden2009a}. Hence, in this section, we are providing two approximations:
\begin{enumerate*}[label=\roman*)]
	\item due to the asymptotic convergence of the innovation filter (Thm.~\ref{thm:boundErrorCov}), the Kalman gain $\Kfnc_t$ can be approximated by $\Kfnc^{(\tau)}_t$, which depends only on $p_{t-\tau},\ldots,p_t$; and
	\item in some cases, by sacrificing state minimally, the approximate Kalman gain with dynamic, rational dependence on the scheduling signal can be transformed to an approximate Kalman gain with static, affine dependence (Sec.~\ref{subsec:staticAffKalman}).
\end{enumerate*}

To start with the first approximation, thm.~\ref{thm:boundErrorCov} highlights that the covariance matrix $\Pfnc_{t\vert t-1}$ can be arbitrary well approximated by only taking the scheduling signal $p$ from $p_{t-\tau},\ldots,p_t$ into account, e.g., ``fading memory'' of the innovation recursions. The approximation error is upper bounded, as given in~\eqref{eq-thm:maxCovarianceBound}, and decays to zero if $\tau\rightarrow\infty$. Furthermore, the covariance matrix~\eqref{eq-lem:innErrorCov} is not implicitly dependent on the Kalman gain~\eqref{eq-lem:kalmanFilter}; however, any approximation of $\Pfnc$ will lead to an approximation of $\Kfnc$. As $\Kfnc$ is a rational function, any approximation of $\Pfnc$ will result in a unique relation in $\Kfnc$ (up to co-primness of the nominator and denominator).

\begin{conj} \label{conj:truncKalman}
\begin{subequations}
Let us consider that there exists a $\tau>0$, $\hat \Pfnc^{\scriptstyle(\tau)}_{t\vert t-1}$ as constructed in Thm.~\ref{thm:boundErrorCov}, and let the associated gain $\Kfnc_t^{\scriptstyle(\tau)}$ be given by~\eqref{eq-lem:kalmanFilter} using $\hat \Pfnc^{\scriptstyle(\tau)}_{t\vert t-1}$. Then the Kalman gain can be decomposed as %\vspace{-3mm}
\begin{equation}
\Kfnc_t = \Kfnc_t^{(\tau)} + \Rfnc_t^{(\tau)},
\end{equation} %\vskip -7mm \noindent
where $\Rfnc_t^{\scriptstyle(\tau)}$ is a rational matrix function in $p_t,p_{t-1},\ldots$. In addition, if $\tau\rightarrow\infty$ then $\Rfnc_t^{\scriptstyle(\tau)}\rightarrow0$ and \vspace{-3mm}
\begin{equation}
\big\Vert \Rfnc_t^{(\tau)} \big\Vert_2 > \big\Vert \Rfnc_t^{(\tau+1)}\big\Vert_2,
\end{equation} %\vskip -7mm \noindent
where $\Rfnc_t^{\scriptstyle(\tau+1)}$ is the remainder term w.r.t. $\Kfnc_t^{\scriptstyle(\tau+1)}$ and $\Kfnc_t^{\scriptstyle(\tau+1)}$ is constructed by using $\hat \Pfnc^{\scriptstyle(\tau+1)}_{t\vert t-1}\!$. \hfill $\square$
\end{subequations}
\end{conj} %\vskip -3mm \noindent %\vskip -7mm \noindent
%\begin{proof}
%\pepijn{Should here be a proof?} \hfill $\blacksquare$
%\end{proof} \vskip -4mm \noindent
%
Conj.~\ref{conj:truncKalman} highlights that $\Kfnc_t$ can be approximated by $\Kfnc_t^{\scriptstyle(\tau)}$, which depends only on $p_{t-\tau},\ldots,p_t$. This truncation can be made arbitrarily accurate by choosing an appropriate $\tau$, i.e., $\big\| \Rfnc_t^{\scriptstyle(\tau+1)} \big\|_2 \ll \big\| \Kfnc_t \big\|_2$.

% ####################################################
% Subsec: LPV-SS innovation from
% ####################################################

\subsection{Static, affine Kalman gain} \label{subsec:staticAffKalman}

A popular choice in LPV-SS identification is to identify an LPV-SS innovation form with a static and affine Kalman filter matrix (e.g., see~\cite{Felici2006,Wingerden2009a}), similarly parametrized as~\eqref{eq:sysMatrices}. Under the assumption that the Kalman filter function can be arbitrarily well approximated by $\Kfnc_t^{\scriptstyle(\tau)}$, the dynamic, rational dependence on the scheduling signal may, in some cases, be transformed into a static, affine LPV-SS representation by adding states, i.e., increasing $\NX$:
\begin{exmp} \label{exmp:dynamic2static}
\begin{subequations}
Consider the following LPV representation
\begin{equation}
y_t = -p_ty_{t-1} + p_tu_{t-1}+e_{t}+e_{t-1}.\label{eq-exmp:LPVIOnoise}
\end{equation}
The state minimal LPV-SS realization of~\eqref{eq-exmp:LPVIOnoise} is
\begin{equation}
x_{t+1}\!=\! -p_tx_t \!+\!  u_t \!+\! \frac{1-p_{t+1}}{p_{t+1}}e_t,\hspace{0.5cm} y_t\! =\! p_tx_t\!+\!e_t,\label{eq-exmp:dynamicLPVSSnoise}
\end{equation}
which is rationally and dynamically dependent on the scheduling parameter $p$. Note, it can be shown that there exists no state transformation (not even $p$-dependent) which can turn~\eqref{eq-exmp:dynamicLPVSSnoise} into~\eqref{eq:SSrepInnState}-\eqref{eq:SSrepInnOut} with static, affine depend $\Kfnc$ on $p$, i.e., $\Kfnc_t=\Kfnc(p_t)$, and keep the minimal state dimension $\NX=1$~\cite[Def. 3.29]{Toth2010a}. However, the transformation into a static, affine form can be done by introducing an additional state as
\begin{equation}
\begin{aligned}
\breve x_{t+1} &\!=\! \left[\begin{array}{cc} -p_t  & 1 \\ 0& 0 \end{array}\right]\!\breve x_t + \left[\begin{array}{cc} -1 & 1 \\0 & 1 \end{array} \right] \left[\begin{array}{c} u_t \\ e_t \end{array} \right]\!,\\
 y_t &\!=\! \left[\begin{array}{cc} -p_t  & 1 \end{array}\hspace{0.05cm}\right]\!\breve x_t+e_t. 
\end{aligned}\label{eq-exmp:staticLPVSSnoise}
\end{equation} 
\end{subequations}
\end{exmp}
Ex.~\ref{exmp:dynamic2static} shows the elimination of dynamic, rational dependence by sacrificing state minimality\footnote{Comparable phenomena can be observed in the LTI case. If it is assumed that $\Sv=0$, however, for the underlying system $\Sv\neq0$, then an increase of the state dimension is also evident~\cite{Anderson1979}.}. Hence, as many LPV-SS identification methods estimate a static, affine functional relation on the scheduling signal~\cite{Felici2006,Wingerden2009a}, the rank relieving property of subspace methods is lost, as additional states are added to preserve the static, affine dependency. As a conclusion, in the LPV case, the Kalman gain $\Kfnc_t$~\eqref{eq-lem:kalmanFilter} should have rational and dynamic dependency on the scheduling signal to enjoy general noise modelling capabilities (Lem.~\ref{lem:inn}) and minimality of the state dimension. However, in practice, we need to restrict overparameterization to reduce complexity of the estimation method and variance of the model estimates. Hence, the above given analysis is important to understand the trade-off behind these choices.

% ############################################### %
% ----------------------------------------------- %
% 		Conclusion
% ----------------------------------------------- %
% ############################################### %

\section{Conclusion}
We have shown that the innovation form~\eqref{eq:SSrepInn} should have a Kalman gain with rational and dynamic dependence on the scheduling signal to represent general noise. However, this function can be approximated by truncating the dynamic dependency. Using this truncation, for some cases, an equivalent LPV-SS representation with affine and static dependency on the scheduling signal can be found by including additional states, resulting in a non-state minimal system. %Consequently, current LPV subspace identification schemes might loose their rank revealing properties.

\appendix[Innovation representation]
% or


%\appendices
%% you can choose not to have a title for an appendix
%% if you want by leaving the argument blank
%\section{}

The idea of \textit{the innovation process $\boldsymbol\xi_t$ is such that $\boldsymbol\xi_t$ consists of that part of $\mathbf y_t$ not carried in $\mathbf y_{t-1},\mathbf y_{t-2},\ldots$}~\cite{Anderson1979}, i.e.,
\begin{equation} \label{eq-app:sepModel}
\boldsymbol \xi_t = \mathbf y_t - \expct^*\{\mathbf y_t~\vert~Y_{t-1}\},
\end{equation}
where $\expct^*\{\cdot\}$ is the minimum variance estimator and $Y_{t-1}$ indicates the set of observations $\{\mathbf y_{t-1},\ldots,\mathbf y_0\}$. The signal $y_t$ generated by~\eqref{eq:SSrep} is a sequence of Gaussian random variables as $u_t$ is known exactly. Hence, the output signal $y$ is split into a `deterministic' part of $\mathbf y_t$ as $\mathbf{\check y}_t=\expct^*\{\mathbf y_t~\vert~Y_{t-1}\}$ and a white noise $\boldsymbol \xi_t$ with Gaussian distribution. The variables $\mathbf{\check y}_t$ and $\boldsymbol \xi_t$ are uncorrelated, i.e., $\expct\{\mathbf{\check y_i}\boldsymbol \xi_i^\top\}=0$ for $i=0,\ldots,t$ because of the orthogonality property of the minimum variance estimator. Without loss of generality, we assume that $\boldsymbol\xi_0 = \mathbf y_0 - \expct^*\{\mathbf y_0\}$. Hence, as the initial condition is known, there exists a causal filter from $\mathbf y_0,\ldots,\mathbf y_t$ to $\boldsymbol\xi_t$ by writing out~\eqref{eq:SSrep}. The other way around, i.e., that $\mathbf y_t$ depends on $\boldsymbol\xi_0,\ldots,\boldsymbol\xi_t$, can be shown in a recursive way~\cite{Anderson1979}. Therefore, the dataset $\mathbf y_0,\ldots,\mathbf y_t$ and $\xi_0,\ldots,\xi_t$ are uniquely related to each other and the following holds
\begin{equation} \label{eq-app:eqalDist}
\expct\{\mathbf y_t~\vert~\mathbf y_0,\ldots,\mathbf y_{t-1}\} = \expct\{\mathbf y_t~\vert~\boldsymbol\xi_0,\ldots,\boldsymbol\xi_{t-1}\}.
\end{equation}
In addition, for any variable $\mathbf x_t$ which has a joint Gaussian distribution with $\mathbf y_t$ it holds that
\begin{equation} \label{eq-app:eqalDistExtVar}
\mathbf{\check x}_t = \expct\{\mathbf x_t~\vert~\mathbf y_0,\ldots,\mathbf y_{t-1}\} = \expct\{\mathbf x_t~\vert~\boldsymbol\xi_0,\ldots,\boldsymbol\xi_{t-1}\}.
\end{equation}
Substituting~\eqref{eq-app:eqalDist} and~\eqref{eq-app:eqalDistExtVar} into~\eqref{eq-app:sepModel} and taking the output equation relation~\eqref{eq:SSrepInnOut} into account, gives
\begin{equation}
\mathbf y_t = \Cfnc(p_t)\mathbf{\check x}_t + \Dfnc(p_t)u_t+\boldsymbol\xi_t. \label{eq-app:SSrepInnOut}
\end{equation}
We will assume that the initial state $x_0=0$ is known\footnote{This proof can be extended to $x_0\in\mathcal{N}(0,P_0)$. However, it involves additional constraints to ensure that the noise sequences $w$ and $v$ can be causally computed from $y$, see~\cite[Theorem 3.4, Ch. 9]{Anderson1979}. For simplicity, this case will not be considered.}. %In addition, the state estimate $\check x$ is a Gaussian random variable, hence, for the minimum variance estimate it holds that $\expct^*\{x_t~\vert~\xi_0,\ldots,\xi_{t-1}\}=\expct\{x_t~\vert~\xi_0,\ldots,\xi_{t-1}\}$.

As $\boldsymbol\xi_0,\ldots,\boldsymbol\xi_{t+1}$ are mutually uncorrelated, the conditional expectation~\eqref{eq-app:eqalDistExtVar} can be split up, e.g., see~\cite[Theorem 2.4, Ch. 5]{Anderson1979}, and combined with~\eqref{eq:SSrepInnState}, which gives
\begin{align}
\mathbf{\check x}_{t+1}\! &=\! \expct\{ \mathbf x_{t+1} ~\vert~\boldsymbol\xi_0,\ldots,\boldsymbol\xi_{t-1} \}\! +\! \expct\{ \mathbf x_{t+1} ~\vert~\boldsymbol\xi_t \}\!-\!\expct\{ \mathbf x_{t+1}\}\!, \nonumber \\[1.5mm]
			  \! &=\!  \Afnc(p_t)\mathbf{\check x}_t\! +\! \Bfnc(p_t)u_t\! +\! \expct\{ \mathbf x_{t+1} \!~\vert~\!\boldsymbol\xi_t \}\!-\!\expct\{\mathbf  x_{t+1}\}\!. \label{eq-app:stateEstExpc}
\end{align}
Note that $\mathbf x_{t+1}$ is uncorrelated with $\boldsymbol\xi_{t-i}$ for $i>0$. 
The state $\mathbf x_{t+1}$ and $\boldsymbol\xi_t$ are jointly Gaussian distributed variables, hence
\begin{multline} \label{eq-app:minVarEst}
\expct\{\mathbf x_{t+1} ~\vert~\boldsymbol\xi_t \} = \expct\{\mathbf x_{t+1}\} + \\ \cov\left[\mathbf x_{t+1},\boldsymbol\xi_t\right]\var\left[\boldsymbol\xi_t\right]^{-1}\left(\boldsymbol\xi_t-\expct\{\boldsymbol\xi_t\}\right),
\end{multline}
by using the minimum variance estimator property, e.g., see~\cite[Theorem 2.1, Ch. 5]{Anderson1979}. Define the error of the state estimate by $\mathbf{\tilde x}_t=\mathbf x_t-\mathbf{\check x}_t$. To compute $\cov\left[\mathbf x_{t+1},\boldsymbol\xi_t\right]$, see:
\begin{multline} \label{eq-app:covStateNoise}
\cov\left[\mathbf x_{t+1},\boldsymbol\xi_t\right] = \cov\left[\mathbf x_{t+1},\Cfnc(p_t)\mathbf{\tilde x}_t+\Hfnc(p_t)\mathbf v_t\right] \\ 
\!=\!\expct\!\left\{\! \left[\Afnc(p_t)\!\left(\mathbf x_t\!-\! \expct\{\mathbf x_t\}\! \right)\! +\! \Gfnc(p_t) \mathbf w_t \right]\! \left[ \Cfnc(p_t)\mathbf{\tilde x}_t\!+\!\Hfnc(p_t)\mathbf v_t  \right]^\top \!\right\} \\
= \Afnc(p_t)\Pfnc_{t\vert t-1}\Cfnc^\top\!\!(p_t) + \Gfnc(p_t)\Sv\Hfnc^\top\!\!(p_t),
\end{multline}
where $\Pfnc_{t\vert t-1} = \var\left[\mathbf{\tilde x}_t \right]$ is the a priori state error covariance. To compute $\var\left[ \boldsymbol\xi_t \right]$, note that~\eqref{eq:SSrepOut} and~\eqref{eq-app:SSrepInnOut} are equal in $\mathbf y_t$, hence, by using the transitive property of equality, the variance of $\boldsymbol\xi_t$ is given as
\begin{align}
\Omega_t &=\var\left[ \boldsymbol\xi_t \right] = \var\left[ \Cfnc(p_t)\mathbf{\tilde x}_t + \Hfnc(p_t)\mathbf v_t \right] \nonumber \\
&= \Cfnc(p_t)\Pfnc_{t\vert t-1}\Cfnc^\top\!\!(p_t) + \Hfnc(p_t)\Rv\Hfnc^\top\!\!(p_t). \label{eq-app:noiseCov}
\end{align}
Substituting~\eqref{eq-app:minVarEst},~\eqref{eq-app:covStateNoise}, and~\eqref{eq-app:noiseCov} in~\eqref{eq-app:stateEstExpc} gives
\begin{subequations}
\begin{align}
\mathbf{\check x}_{t+1}\! &= \!\Afnc(p_t)\mathbf{\check x}_t+\Bfnc(p_t)u_t+\Kfnc_t\boldsymbol\xi_t, \label{eq-app:SSrepInnState}\\
\Kfnc_t \!&= \!\left[ \Afnc(p_t) \Pfnc_{t\vert t-1} \Cfnc^\top\!\!(p_t)\! +\! \Gfnc(p_t)\Sv\Hfnc^\top\!\!(p_t)\right]\! \Omega_t^{-1}\!.\label{eq-app:kalmanFilter}
\end{align}
\end{subequations}
Finally, the a priori state error covariance $\Pfnc_{t\vert t-1}$ should be found. Subtracting~\eqref{eq-app:SSrepInnState} from~\eqref{eq:SSrepState} gives
\begin{multline}
\mathbf{\tilde x}_{t+1} = \Afnc(p_t)\mathbf{\tilde x}_t + \Gfnc(p_t)\mathbf w_t-\Kfnc_t\boldsymbol\xi_t  \\
\!=\!\left[\Afnc(p_t)\!-\!\Kfnc_t\Cfnc(p_t)\right]\mathbf{\tilde x}_t \!+\! \Gfnc(p_t)\mathbf w_t\!-\!\Kfnc_t\Hfnc(p_t)\mathbf v_t.
\end{multline}
Then
\begin{align}
&\Pfnc_{t+1\vert t} = \left[\Afnc(p_t)-\Kfnc_t\Cfnc(p_t)\right]\Pfnc_{t\vert t-1}\left[\Afnc^\top\!\!(p_t)-\Cfnc^\top\!\!(p_t)\Kfnc^\top_t\right] \nonumber\\
&\qquad+\Gfnc(p_t)\Qv\Gfnc^\top\!\!(p_t)+\Kfnc_t\Hfnc(p_t)\Rv\Hfnc^\top\!\!(p_t)\Kfnc^\top_t\nonumber\\
&\qquad-\Gfnc(p_t)\Sv\Hfnc^\top\!\!(p_t)\Kfnc^\top_t - \Kfnc_t\Hfnc(p_t)\Sv^\top\Gfnc^\top\!\!(p_t) \nonumber\\
&=\Afnc(p_t)\Pfnc_{t\vert t-1}\Afnc^\top\!\!(p_t)+\Gfnc(p_t)\Qv\Gfnc^\top\!\!(p_t)+\Kfnc_t\Omega_t\Kfnc_t \nonumber\\
&\qquad-\Kfnc_t[\Cfnc(p_t)\Pfnc_{t\vert t-1}\Afnc^\top\!\!(p_t) +\Hfnc(p_t)\Sv^\top \Gfnc^\top\!\!(p_t)]\nonumber\\
&\qquad-[\Afnc(p_t)\Pfnc_{t\vert t-1}\Cfnc^\top\!\!(p_t) +\Gfnc(p_t)\Sv \Hfnc^\top\!\!(p_t)] \Kfnc^\top_t,\label{eq-app:errCov}
\end{align}
Combining~\eqref{eq-app:kalmanFilter} and~\eqref{eq-app:errCov} gives~\eqref{eq-lem:innErrorCov}.

% \section*{Acknowledgment}

\ifCLASSOPTIONcaptionsoff
  \newpage
\fi

\bibliographystyle{IEEEtran}
\bibliography{library.bib}           % and a bib file to produce the 
\end{document}